\documentclass{article}[11pt]
\usepackage{amssymb}
\usepackage{amsmath,amsthm}
\usepackage{color}
\newcommand{\clocks}{\mathcal{X}} 
\newcommand{\true}{\texttt{true}} 
\newcommand{\N}{\mathbb{N}}
\newcommand{\RP}{\mathbb{R}_{\ge 0}}
\newcommand{\val}{\nu} 
\newcommand{\guard}{\varphi} 
\newcommand{\reset}{\lambda}
\newcommand{\zeroval}{\boldsymbol{0}} 
\newcommand{\A}{\mathcal{A}} 
\newcommand{\B}{\mathcal{B}} 
\newcommand{\locs}{L} 
\newcommand{\edges}{E}
\newcommand{\tuple}[1]{\langle #1 \rangle}

\newcommand{\loc}{\ell} 
\newcommand{\sto}[1]{\stackrel{#1}{\Longrightarrow}}

\newtheorem{theorem}{Theorem}
\newtheorem{proposition}[theorem]{Proposition}

\advance \topmargin by -\headheight
\advance \topmargin by -\headsep
\evensidemargin \oddsidemargin
\begin{document}
\title{Effective Definability of the Reachability Relation in Timed
  Automata} 

\author{Martin Fr\"anzle\\
University Of Oldenburg
\and Karin Quaas \\
Universit\"at Leipzig
\and
Mahsa Shirmohammadi\\
CNRS \& IRIF
\and
 James Worrell\\
University of Oxford}

\date{}
\maketitle
\begin{abstract}
We give a new proof of the result of Comon and Jurski that the binary reachability relation of a timed automaton 
is definable in linear arithmetic.
\end{abstract}

%
%

\section{Introduction}
Comon and Jurski~\cite{ComonJ99,ComonJ99-TR} showed that the binary
reachability relation of a given timed automaton is effectively
definable by a first-order formula of linear arithmetic over the reals
augmented with a unary predicate denoting the integers.  The proof of
this result, given in~\cite{ComonJ99-TR}, is based on a syntactic
transformation of arbitrary timed automata into equivalent timed
automata satisfying a certain structural restriction, called flatness.
The proof is relatively long and technical (running to over 40 pages)
and there have been a number of subsequent attempts to both generalise
the result and simplify its
proof~\cite{ClementeL18,Dang03,Dima02,QuaasSW17}.  The present note is
a development of~\cite[Sections III and IV]{QuaasSW17} and further
simplifies the proof of Comon and Jurski's result therein.  In
particular, we avoid many technicalities of~\cite{QuaasSW17} by
employing a simple ``clock memorisation'' trick to reduce computation of the
binary reachability relation of a timed automaton to computation of
the set of configurations reachable from a given location starting
with the all-zeros clock valuation.  We show how to recover the latter
set of configurations as the commutative image of a certain regular
language accepted by a variant of Alur and Dill's region automaton
(cf.~\cite{AlurD94}).

\section{Definitions and Main Result}
Given a set~$\clocks=\{x_1,\ldots,x_n\}$ of \emph{clocks}, the set
$\Phi(\clocks)$ of \emph{clock constraints} is generated by the
grammar
\[ \varphi ::= \true \mid x<k \,\mid\, x = k \,\mid\, x>k \,\mid\,
  \varphi \wedge \varphi \, , \] where $k \in \N$ 
and $x\in \clocks$.  A \emph{clock valuation} is a
mapping~$\val: \clocks \to \RP$, where $\RP$ is the set of
non-negative real numbers.  Denote by $\RP^{\clocks}$ the set of all
clock valuations.  We write $\val\models\varphi$ to denote
that~$\val \in \RP^{\clocks}$ satisfies the constraint $\guard$.
We denote by $\zeroval$ the valuation such
that~$\zeroval(x)=0$ for all $x\in \clocks$.    Given
$t\in\RP$, we let $\val+t$ be the clock valuation such that
$(\val+t)(x)=\val(x)+t$ for all clocks~$x\in\clocks$.  Given
$\reset\subseteq\clocks$, let $\val[\reset\leftarrow 0]$ be the clock
valuation such that $\val[\reset\leftarrow 0](x)=0$ if $x\in\reset$,
and $\val[\reset\leftarrow 0](x)=\val(x)$ if $x\not\in\reset$.  

A \emph{1-bounded zone} $Z \subseteq [0,1]^\clocks$ is a set of clock valuations
 that is defined by a conjunction of
difference constraints $x_i \sim c$ and $x_i - x_j \sim c$, where $c \in \{-1,0,+1\}$,
${\sim} \in \{ <,= \}$, $i,j \in \{1,\ldots,n\}$.  We write
$\mathcal{Z}_1(\clocks)$ for the set of 1-bounded zones.  Given a
1-bounded zone $Z$ and $\lambda\subseteq\clocks$, the following are 
also 1-bounded zones (see, e.g.,~\cite{Bouyer03}):
\begin{eqnarray*}
Z[\lambda \leftarrow 0] &:=& \{ \nu[\lambda\leftarrow 0] : \nu \in Z\}\\
\overrightarrow{Z} &:=& \{ \nu+t : \nu \in Z,t\geq 0\} \cap [0,1]^\clocks \, .
\end{eqnarray*}

A \emph{timed automaton} is a tuple $\A=\tuple{\locs,\clocks,\edges}$,
where~$\locs$ is a finite set of \emph{locations}, $\clocks$ is a
finite set of \emph{clocks}, and
$\edges\subseteq \locs\times \Phi(\clocks)\times 2^\clocks\times
\locs$ is the set of \emph{edges}.  A \emph{configuration} of $\A$ is
a pair $\tuple{\loc,\val}$ consisting of a location~$\loc$ and a clock
valuation~$\val$.  Such a timed automaton $\A$ induces a 
ternary \emph{transition relation} 
\[ {\sto{}} \subseteq (\locs\times \RP^{\clocks}) \times \mathbb{R}
  \times (\locs\times\RP^{\clocks}) \] on the set of configurations as
follows. Given configurations~$\tuple{\loc,\val}$ and
$\tuple{\loc',\val'}$, we postulate:
\begin{itemize}
\item a delay transition~$\tuple{\ell,\val}\sto{d}\tuple{\ell',\val'}$ for
    some $d\geq 0$, if~$\val'=\val+d$ and $\ell=\ell'$;
        \item a discrete transition~$\tuple{\ell,\val}\sto{0}\tuple{\ell',\val'}$, if there
is an edge $\tuple{\ell,\varphi,\reset,\ell'}$ of $\A$
        such that $\val\models\guard$ and $\val'=\val[\reset \leftarrow 0]$.
\end{itemize}
A \emph{run}~$q_0 \sto{d_1} q_1 \sto{d_2} q_2 \sto{d_3} \ldots \sto{d_m}
q_m$ of $\A$ is a finite sequence of delay and discrete transitions.

Let $\mathcal{L}$ denote the set of first-order formulas over the
structure $\mathcal{R}=(\mathbb{R},\mathbb{Z}(\cdot),0,1,+,\leq)$,
where $\mathbb{R}$ is the universe and $\mathbb{Z}(\cdot)$ is a unary
predicate denoting the set of integers.  We call $\mathcal{L}$ the
language of \emph{mixed linear arithmetic}.  This language subsumes
both Presburger arithmetic (linear arithmetic over $\mathbb{Z}$) and
linear arithmetic over $\mathbb{R}$.  It is shown in~\cite[Section
  3]{BoigelotJW05} how to rewrite a given $\mathcal{L}$-sentence in
polynomial time into an equivalent Boolean formula whose atoms are
either sentences of real linear arithmetic or sentences of Presburger
arithmetic.  From known complexity bounds for the latter two
theories~\cite{Berman80}, it follows that deciding the truth of
$\mathcal{L}$-sentences can be carried out by an alternating Turing
machine running in time $2^{2^{O(n)}}$ with $n$ alternations (i.e.,
the same complexity as Presburger arithmetic).

The main contribution of the present paper is to prove the following result.
\begin{theorem}
Given a timed automaton $\A$ with $n$ clock variables and
locations~$\ell_0,\ell$ of $\A$, we can compute an
$\mathcal{L}$-formula
$\varphi^{\A}_{\ell_0,\ell}(x_1,\ldots,x_n,y_1,\ldots,y_n)$ such that
there is a run of $\A$ from configuration~$\tuple{ \ell_0,\val_0}$ to
configuration~$\tuple{ \ell,\val}$ iff $\mathcal{R} \models
\varphi^{\A}_{\loc_0,\loc}[\val_0,\val]$.
\label{thm:main}
\end{theorem}

\section{Proofs}
\label{sec:proofs}
\subsection{Clock Memorisation}
\label{subsec:mem}
We describe a simple trick that reduces the problem of computing the
binary reachability relation on a given timed automaton $\A$ to that
of computing the set of configurations reachable from a fixed initial
configuration in a derived automaton $\B$.  The idea is that $\B$
starts from the zero clock valuation, guesses an initial clock
valuation $\nu_0$ of $\A$, and then simulates a computation of $\A$
while ``remembering'' $\nu_0$ as a set of differences between the
values of some fresh clocks.  Formally we derive
Theorem~\ref{thm:main} from the following result, which we prove later on.
\begin{proposition}
  Given a timed automaton $\A$ with $n$ clocks and locations $\loc_0,\loc$ 
  of $\A$, we can compute an $\mathcal{L}$-formula
  $\psi^{\A}_{\loc_0,\loc}(x_1,\ldots,x_n)$ such that there is a run in $\A$
  from $\tuple{\loc_0,\boldsymbol{0}}$ to $\tuple{\loc,\nu}$ if and
  only if $\mathcal{R} \models \psi^{\A}_{\loc_0,\loc}[\nu]$.
\label{prop:main2}
\end{proposition}

\begin{proof}[Proof of Theorem~\ref{thm:main}]
Given a timed automaton
$\A=\tuple{\locs,\clocks,\edges}$ with a
distinguished location $\loc_0\in\locs$, we define a
new timed automaton $\B =
\tuple{\locs',\clocks',\edges'}$ as follows.  The
set of locations is $\locs'=\locs \cup \{ \loc'_0 \}$,
where $\loc'_0\not\in \locs$ is a distinguished 
location in $\B$.  The set of clocks is $\clocks' = \{ x, x' : x
\in \clocks \} \cup \{ z\}$, where $z\not\in\clocks$,
that is, $\B$ has two copies of each clock of $\A$ plus an extra
``reference clock'' $z$.  We obtain $\edges'$ by adding the following edges to
$\edges$:
for each $x\in\clocks$ we have an edge
$\tuple{\loc'_0,\mathbf{true},\{x,x'\},\loc'_0}$ in $\edges'$
(that is, a selfloop on $\loc'_0$ that resets both copies of clock
$x$); we also have a single additional edge
$\tuple{\loc'_0,\mathbf{true},\{ z \},\loc_0}$ in $\edges'$.
Notice that once $\B$ leaves $\loc'_0$ then neither the reference
clock~$z$ nor any clock in the set $\{x'\mid x\in\clocks\}$ is
reset---intuitively when $\B$ exits $\loc'_0$, the value of clock~$x \in \clocks$ is
stored in the difference of $x'$ and $z$.

Observe that for all $\ell \in \locs$ there is a run from
$\tuple{\loc_0,\nu_0}$ to $\tuple{\loc,\nu}$ in $\A$ if and only if
there is a run in $\B$ from $\tuple{ \loc'_0, \boldsymbol{0}}$ to
$\tuple{\loc,\nu'}$ such that $\nu'(x) = \nu(x)$ and $\nu'(x') -
\nu'(z) = \nu_0(x)$ for all $x \in \clocks$.  In particular, a run of
$\A$ from $\tuple{\loc_0,\nu_0}$ to $\tuple{\loc,\nu}$ can be
simulated in $\B$ as follows.  Automaton $\B$ starts in configuration
$\tuple{\loc_0',\boldsymbol 0}$; by taking selfloops on $\loc_0'$ and
then the edge $\tuple{\loc'_0,\mathbf{true},\{ z \},\loc_0}$, $\B$ may
reach a configuration $\tuple{\loc_0,\nu''}$ such that
$\nu''(x)=\nu''(x')=\nu_0(x)$ for all $x \in \clocks$, and $\nu''(z)=0$; then $\B$ directly
simulates the given run of $\A$ (without resetting the new clocks
$x'$, $x \in \clocks$, and $z$).

Let $\psi^{\B}_{\loc'_0,\loc}(\boldsymbol x,\boldsymbol x',z)$ be the formula
obtained in Proposition~\ref{prop:main2} for automaton $\B$.  Then we define
\[ \varphi^{\A}_{\loc_0,\loc}(\boldsymbol x,\boldsymbol y) :=
   \exists \boldsymbol x' \exists z \, . \, (\boldsymbol x = \boldsymbol x'-z\boldsymbol{1} \wedge 
\psi^{\B}_{\loc'_0,\loc}(\boldsymbol y,\boldsymbol x',z)) \, . \]
\end{proof}

\subsection{A Discrete-Time Automaton}
Given a timed automaton $\A=\tuple{\locs,\clocks,\edges}$, we define an (untimed, infinite-state)
nondeterministic automaton $R(\A)$.  The set of states of $R(\A)$ is
\[ Q=\locs \times \mathbb{N}^\clocks \times \mathcal{Z}_1(\clocks)
  \times 2^{\clocks} \, .\] Given a state
$\tuple{\loc,\upsilon,Z,\gamma}\in Q$, intuitively $\upsilon$ and $Z$
respectively encode the integer and fractional parts of the clocks of $\A$,
while $\gamma\subseteq \clocks$ is a ``prophecy variable'' denoting
the set of clocks that will be reset at least once in the future.  The alphabet of
$R(\A)$ is ${\mathcal{X}}$ and the set of transitions (which includes
$\varepsilon$-transitions) is as follows:
\begin{enumerate}
\item For each state $\tuple{\loc,\upsilon,Z,\gamma}\in Q$ 
there is a \emph{delay transition}
  $\tuple{\loc,\upsilon,Z,\gamma}
  \stackrel{\varepsilon}{\longrightarrow}
  \tuple{\loc,\upsilon,\overrightarrow{Z},\gamma}$.

\item For each clock $x\in\clocks$ and state
  $\tuple{\loc,\upsilon,Z,\gamma}$ of $R(\A)$ there is a
  \emph{wrapping transition}
  $\tuple{\loc,\upsilon,Z,\gamma} \stackrel{\sigma}{\longrightarrow}
  \tuple{\loc,\upsilon',Z',\gamma}$ where
  $Z':=(Z\cap [\![x=1]\!])[x \leftarrow 0]$,
  $\upsilon'=\upsilon[x\leftarrow x+1]$, and $\sigma=\varepsilon$ if
  $x \in \gamma$ but otherwise $\sigma=x$.  (Intuitively a wrapping
  transition for clock~$x$ has label $\varepsilon$ if $x \in \gamma$
  since $x$ will be reset again in the future.)

\item Each edge $\tuple{\loc,\varphi,\lambda,\loc'}$ of $\A$ yields a
  transition $\tuple{\loc,\upsilon,Z,\gamma}
  \stackrel{\varepsilon}{\longrightarrow}
  \tuple{\loc',\upsilon',Z',\gamma'}$ of $R(\A)$, where
  $\upsilon'=\upsilon[\lambda \leftarrow 0]$,   $Z':=\{ \nu \in Z :
  \upsilon+\nu \models \varphi \}[\lambda \leftarrow 0]$, and  $\gamma'\cup\lambda=\gamma$.
(Intuitively, $\gamma'$ is obtained from $\gamma$ by guessing some clocks that will not be reset again
and removing them from $\gamma$.)

\end{enumerate}
Given states $q,q' \in Q$, we write
$q\stackrel{w}{\longrightarrow} q'$ if there is a run in $R(\A)$ from
$q$ to $q'$ on word $w\in\clocks^*$.  We furthermore write
$q\longrightarrow^* q'$ if there exists a run from $q$ to $q'$.


\begin{proposition}
  Automaton $\A$ has a run from 
  $\tuple{\loc_0,\boldsymbol{0}}$ to $\tuple{\loc,\nu}$
  along which the set of clocks that are reset is
  $\gamma\subseteq\clocks$ if and only if $R(\A)$ has a run from
  $\tuple{\loc_0,\boldsymbol{0},\{\boldsymbol{0}\},\gamma}$ to
  $\tuple{\loc,\upsilon,Z,\emptyset}$ for some $\upsilon\in\mathbb{N}^\clocks$ and
  $Z\in\mathcal{Z}_1(\clocks)$ such that $\nu-\upsilon \in Z$.
\label{prop:discrete}
\end{proposition}
\begin{proof}
 For the ``if'' direction, suppose that
\[ \tuple{\loc_0,\upsilon^{(0)},Z_0,\gamma_0} \longrightarrow                                                                                                 \tuple{\loc_1,\upsilon^{(1)},Z_1,\gamma_1} \longrightarrow                                                                                           
   \ldots \longrightarrow                                                                                                                                  
   \tuple{\loc_k,\upsilon^{(k)},Z_k,\gamma_k} \]
is a run of $R(\A)$ with $\upsilon^{(0)}=\boldsymbol{0}$,
$Z_0=\{\boldsymbol{0}\}$, and $\gamma_0=\gamma$.
Given any valuation $\nu^{(k)} \in Z_k$,
we construct a sequence of valuations $\nu^{(0)} \in Z_0,\ldots,\nu^{(k-1)} \in Z_{k-1}$
such that
$\A$ has a run
\[ \tuple{\loc_0,\upsilon^{(0)}+\nu^{(0)}} \Longrightarrow                                                                                                 
   \tuple{\loc_1,\upsilon^{(1)}+\nu^{(1)}} \Longrightarrow                                                                                                 
   \ldots \Longrightarrow                                                                                                                                  
\tuple{\loc_k,\upsilon^{(k)}+\nu^{(k)}} \, . \]
The construction of $\nu^{(j)}$ is by backward induction on $j$.  The
base step, valuation $\nu^{(k)}$, is given.  The induction step
divides into three cases according to the type of the transition
\[ \tuple{\loc_{j-1},\upsilon^{(j-1)},Z_{j-1},\gamma_{j-1}}                            \longrightarrow \tuple{\loc_{j},\upsilon^{(j)},Z_j,\gamma_j} \, .\]
\begin{itemize}
\item Delay transition.  We have
  $Z_j = \overrightarrow{Z_{j-1}}$,
  $\loc_j = \loc_{j-1}$, and $\upsilon^{(j)}=\upsilon^{(j-1)}$.  Thus
  we can pick $\nu^{(j-1)} \in {Z_{j-1}}$ such that
  $\nu^{(j)}=\nu^{(j-1)}+d$ for some $d\geq 0$.  Hence there is in $\A$ a
  delay transition
\[ \tuple{\loc_{j-1},\upsilon^{(j-1)}+\nu^{(j-1)}} \stackrel{d}{\Longrightarrow}                                                                            
   \tuple{\loc_j,\upsilon^{(j)}+\nu^{(j)}} \, . \]
\item Wrapping transition.  We have $Z_j = (Z_{j-1} \cap                     [\![x=1]\!])[x \leftarrow 0]$ for some clock $x\in\clocks$.  Thus we can pick
$\nu^{(j-1)} \in {Z_{j-1} \cap [\![x=1]\!]}$ such that
$\nu^{(j)}=\nu^{(j-1)}[x\leftarrow 0]$.  In this case we have
\[ \tuple{\loc_{j-1},\upsilon^{(j-1)}+\nu^{(j-1)}} =                                                                                                        
   \tuple{\loc_j,\upsilon^{(j)}+\nu^{(j)}} \, . \]
\item Discrete transition.  Let the corresponding edge of $\A$ be
  $\tuple{\loc_{j-1},\varphi,\lambda,\loc_j}$.  Then we have
$\upsilon^{(j)} = \upsilon^{(j-1)}[\lambda \gets 0]$ and   
$Z_j = \{ \nu \in Z_{j-1} : \nu+\upsilon^{(j-1)} \models \varphi\} 
[\lambda \leftarrow 0]$.
Choose
  $\nu^{(j-1)} \in Z_{j-1}$ such that 
  $\nu+\upsilon^{(j)} \models \varphi$ and 
  $\nu^{(j)} = \nu^{(j-1)}[\lambda \leftarrow 0]$.  Then there is in $\A$ a
  discrete transition
\[ \tuple{\loc_{j-1},\upsilon^{(j-1)}+\nu^{(j-1)}} \stackrel{0}{\Longrightarrow}                                                                            
   \tuple{\loc_j,\upsilon^{(j)}+\nu^{(j)}} \, . \]
\end{itemize}

We now turn to the  ``only-if'' direction of the proof.
Suppose that we have a run
\[ \tuple{\loc_0,\nu^{(0)}} \stackrel{d_1}{\Longrightarrow}                                                                                                 
   \tuple{\loc_1,\nu^{(1)}} \stackrel{d_2}{\Longrightarrow}                                                                                                 
\ldots \stackrel{d_k}{\Longrightarrow} \tuple{\loc_k,\nu^{(k)}} \]
  of $\A$, where $\nu^{(0)}=\boldsymbol{0}$.  We first
  transform such a run, while keeping the same initial and final
  configurations, by decomposing each delay step into a sequence of
  shorter delays, so that for all $0 \leq j \leq k-1$ and all
  $x\in\clocks$ the open interval $(\nu^{(j)}(x),\nu^{(j+1)}(x))$
  contains no integer.  In other words, we break every delay step at every
  point at which some clock crosses an integer boundary.  We thus
  obtain a corresponding run of $R(\A)$ that
  starts from state
  $\tuple{\loc_0,\upsilon^{(0)},Z_0,\gamma_0}$, where
$Z_0 = \{ \boldsymbol{0} \}$, $\upsilon^{(0)} = \boldsymbol{0}$,
and ends in state $\tuple{\loc_k,\upsilon^{(k)},Z_k,\emptyset}$ such that
$\nu^{(k)} \in \upsilon^{(k)}+Z_k$.

We build such a run of $R(\A)$ by forward induction.  In particular,
we construct a sequence of intermediate states
$\tuple{\loc_i,\upsilon^{(i)},Z_i,\gamma_i}$, $0 \leq i \leq k$, such
that $\nu^{(i)} \in \upsilon^{(i)} + Z_i$ for each such $i$.  Each
discrete transition of $\A$ is simulated by a discrete transition of
$R(\A)$.  A delay transition of $\A$ that ends with set of clocks
$\lambda \subseteq \clocks$ being integer valued is simulated by a
delay transition of $R(\A)$, followed by wrapping transitions for all
$x \in \lambda$.
\end{proof}

\begin{proposition}
  Let $\A=\langle \locs,\clocks,E\rangle$ be a timed automaton with
  distinguished locations $\loc_0,\loc \in L$.  For every
  $Z\in\mathcal{Z}_1(\clocks)$ and $\gamma\subseteq\clocks$ the set
\begin{gather} \left\{ \upsilon \in \mathbb{N}^{\clocks} : \exists w\in\clocks^* \cdot 
  \tuple{\loc_0,\boldsymbol{0},\{\boldsymbol{0}\},\gamma} \stackrel{w}{\longrightarrow}
  \tuple{\loc,\upsilon,Z,\emptyset} \right\}
\label{eq:theset}
\end{gather}
is effectively semilinear.
\label{prop:effective}
\end{proposition}
\begin{proof}
Consider the following language of ``wrapping transitions'' in $R(\A)$:
\[ L_{\mathrm{wrap}}:=\left\{w \in \clocks^* : \exists \upsilon \in  \mathbb{N}^{\clocks} \cdot 
  \tuple{\loc_0,\boldsymbol{0},\{\boldsymbol{0}\},\gamma}
  \stackrel{w}{\longrightarrow} \tuple{\loc,\upsilon,Z,\emptyset}
  \right\} \, . \] Define the function $\pi : \clocks^* \rightarrow
  \mathbb{N}^\clocks$ such that $\pi(w)(x)$ is the number of
  occurrences of letter $x$ in word $w$.  Since visible
  transitions of $R(\A)$ correspond to wrapping transitions of clocks
  that will not be reset any more, the commutative image
  $\pi(L_{\mathrm{wrap}})$ is precisely the set of vectors defined in
  (\ref{eq:theset}).

We claim that $L_{\mathrm{wrap}}$ is regular, which 
suffices to show that the set (\ref{eq:theset}) is semilinear.
  Indeed, denote by
$c_{\max}$ the largest constant appearing in a transition constraint in
$\A$, and consider the equivalence relation on states of $R(\A)$
defined by $(\loc,\upsilon,Z,\gamma) \sim (\loc,\upsilon',Z,\gamma)$
iff for all $x\in \clocks$ either $\upsilon(x)=\upsilon'(x)$ or
$\upsilon(x),\upsilon'(x) \geq c_{\max}$.  Clearly if
$(\loc,\upsilon,Z,\gamma) \sim (\loc,\upsilon',Z,\gamma)$ then for all
$\nu\in Z$ and clock constraints $\varphi$ appearing in $\A$ we have
that $\upsilon+\nu\models \varphi$ iff $\upsilon'+\nu\models \varphi$.
It follows that $\sim$ is a strong bisimulation (that moreover has
finite index).  Taking the quotient of $R(\A)$ by $\sim$ it follows
that $L_{\mathrm{wrap}}$ is accepted by a finite automaton and hence
is regular.
\end{proof}

\begin{proof}[Proof of Proposition~\ref{prop:main2}]
By Proposition~\ref{prop:discrete} we have that 
\[ \left\{ \nu \in \mathbb{R}_{\geq 0} ^\clocks :
  \tuple{\loc_0,\boldsymbol{0}} \Longrightarrow^* \tuple{\loc,\nu} \right\} =
  \bigcup_{Z\in \mathcal{Z}_1(\clocks)}
  \bigcup_{\gamma\subseteq\clocks} \left\{ \upsilon + Z :
  \tuple{\loc_0,\boldsymbol{0},\{\boldsymbol{0}\},\gamma}
  \longrightarrow^* \tuple{\loc,\upsilon,Z,\emptyset} \right\} \] But by
Proposition~\ref{prop:effective} the right-hand expression above is
definable by an $\mathcal{L}$-formula in the sense of
Proposition~\ref{prop:main2}.
\end{proof}

\section{Computational Complexity}
In this section we briefly retrace the proof of Theorem~\ref{thm:main} in
order to give a complexity bound for computing the $\mathcal{L}$-formula
$\varphi_{\ell_0,\ell}$ described therein.  We sketch a proof that
$\varphi_{\ell_0,\ell}$ is an existential formula that can be computed
  in time polynomial in the number of locations of the timed automaton
  and exponential in the number of clocks and bit length of the
  maximum clock constant.

  Working backwards, we start by considering the regular language
  $L_{\mathrm{wrap}}$ in the proof of
  Proposition~\ref{prop:effective}.  From the bound
  $|\mathcal{Z}_1(\mathcal{X})| \leq (2 |\clocks|+1)!\leq 2^{O(|\clocks| \log  |\clocks| )}$ it is
  straightforward that there is a finite automaton accepting language
  $L_{\mathrm{wrap}}$ with number of states bounded by
  $\mathrm{poly}(|L|,c_{\max},2^{|\mathcal{X}| \log |\mathcal{X}|})$
  (where $L$ and $\mathcal{X}$ are as in the statement of
  Proposition~\ref{prop:effective}) and moreover this automaton can be
  computed in time polynomial in its size.

  Lin~\cite{Lin10} shows that the Parikh image of the language of an
  NFA with $m$ states and alphabet size $k$ is described by a
  quantifier-free formula of Presburger arithmetic that can be
  computed in time $2^{O(k^2 \log m)}$.  We may thus refine the
  statement of Proposition~\ref{prop:effective} by specifying that the
  subset of $\mathbb{N}^{\clocks}$ in Equation (\ref{eq:theset}) is
  described by a quantifier-free formula of Presburger arithmetic that
  can be computed in time bounded by
  $\mathrm{poly}(|L|,c_{\max},2^{|\mathcal{X}|^2})$.

  Moving now to the proof of Proposition~\ref{prop:main2}, we see that
  the formula $\psi_{\ell_0,\ell}$ can likewise be computed in time bounded
  by $\mathrm{poly}(|L|,c_{\max},2^{|\mathcal{X}|^2})$.  Constructing
  formula $\psi_{\ell_0,\ell}$ from the Presburger-arithmetic formula
  referred to in Proposition~\ref{prop:effective} requires to
  introduce existentially quantified variables to denote the
  fractional parts of the clock values of the source and target
  configurations; but $\psi_{\ell_0,\ell}$ is otherwise quantifier-free.

  Finally, recall that the formula $\varphi_{\ell_0,\ell}$ in
  Theorem~\ref{thm:main} was obtained from the formula
  $\psi_{\ell_0,\ell}$ in Proposition~\ref{prop:main2} with only a
  polynomial blow up that introduced extra existentially quantified
  variables.

\end{document}